\documentclass{article}
\usepackage{amssymb,amsmath,latexsym,amscd,amsfonts, bbm}  
\usepackage{graphicx}  
\usepackage[usenames]{color}
\usepackage{umoline}
  
\newtheorem{theorem}{Theorem}[section]

\newtheorem{claim}[theorem]{Claim}  
\newtheorem{lem}[theorem]{Lemma}

\newtheorem{fact}[theorem]{Fact}

 \newcommand{\qedsymb}{\hfill{\rule{2mm}{2mm}}}  
 \newenvironment{proof}[1][]{\begin{trivlist}  
 \item[\hspace{\labelsep}{\bf\noindent Proof#1:\/}] 
 }{\qedsymb\end{trivlist}}

\newcommand{\ignore}[1]{}

\newcommand{\norm}[1]{{\| #1 \|}}  
  
\newcommand{\ket}[1]{{ |{#1} \rangle }}  
\newcommand{\bra}[1]{{ \langle {#1} | }}
\newcommand{\braket}[2]{{ \langle {#1} | {#2} \rangle}}
\newcommand{\ketbra}[2]{{ |{#1} \rangle\langle {#2} | }}

\newcommand{\orderof}[1]{\mathcal{O}(#1)} 
 
\newcommand{\EqDef}{\stackrel{\mathrm{def}}{=}}
\newcommand{\Tr}{\mathrm{Tr}}

\newcommand{\Hi}{{\mathcal{H}}}
\newcommand{\gs}{\Pi_{gs}}

\newcommand{\Eq}[1]{Eq.~(\ref{#1})}
\newcommand{\Fig}[1]{Fig.~\ref{#1}}

\newcommand{\Lem}[1]{Lemma~\ref{#1}}

\newcommand{\Sec}[1]{Sec.~\ref{#1}}
\newcommand{\Ref}[1]{Ref.~\cite{#1}}

\newcommand{\App}[1]{Appendix~\ref{#1}}

\newcommand{\Id}{\mathbbm{1}}
\newcommand{\BBC}{\mathbbm{C}}

\begin{document}

\title{Quantum Hamiltonian complexity and the detectability lemma}

\author{Dorit Aharonov\thanks{School of Computer Science and 
  Engineering, The Hebrew University, 
    Jerusalem, Israel}
 \and Itai Arad\thanks{School of Computer Science and 
    Engineering, The Hebrew University, Jerusalem, Israel.} 
 \and Zeph Landau\thanks{Department of Computer Science, UC Berkeley} 
 \and Umesh Vazirani\thanks{Department of Computer Science, UC Berkeley} 
}

\maketitle

\begin{abstract}
  Local Hamiltonians, the central object of study in condensed
  matter physics, are the quantum analogue of CSPs, and ground
  states of Hamiltonians are the quantum analogue of satisfying
  assignments. The major difference between the two is the existence
  of multi-particle entanglement in the ground state, which
  introduces a whole new level of difficulty in tackling questions
  such as quantum PCP, quantum analogues of amplification, etc.

  The Lieb-Robinson bound is a sophisticated analytic tool used in
  condensed matter physics for handling quantum correlations in
  ground states, by bounding the velocity at which disturbances
  propagate in quantum local systems. In this paper we show that the
  detectability lemma (introduced in a different context in
  \Ref{ref:Aha09b}), when viewed from the right perspective, can be
  used in place of the Lieb-Robinson bound for the rich case of
  frustration free Hamiltonians. The advantage of this is that the
  resulting proofs are simpler and more combinatorial, and may be
  generalizable to solve some of the most fundamental questions in
  Hamiltonian complexity. Additionally, we give an alternative proof
  of the detectability lemma, which is not only simple and intuitive,
  but also removes a key restriction in the original statement,
  making it more suitable for this new context.

  Specifically, we use the detectability lemma to give a simpler
  proof of Hastings' seminal 1D area law \cite{ref:Has07} for
  frustration-free systems. Proving the area law for two and higher
  dimensions is one of the most important open questions in
  Hamiltonian complexity, and the combinatorial nature of the
  detectability lemma based proof and the resulting simplification
  holds out hope for a possible generalization. We also provide a
  one page proof of Hastings' proof that the correlations in the
  ground states of gapped Hamiltonians decay exponentially with the
  distance (once again, restricted to frustration-free systems). We
  argue that the detectability lemma in this form constitutes a
  basic tool for the study of local Hamiltonians and their ground
  states from a computational point of view.

\end{abstract}  

\newpage
\section{Introduction}

Local Hamiltonians and ground states, the central object of study of
condensed matter physics, are the quantum analogues of the central
objects of study in computational complexity: constraint
satisfaction problems (CSP) and their satisfying assignments.  This
connection, which ties together two seemingly very different areas,
is the starting point for the emergence of the new field, Quantum
Hamiltonian Complexity, in which properties of local Hamiltonians
and ground states are being studied from a computational complexity
point of view. Over the past few years, this direction has shed
exciting new insights into quantum information theory as well as
into quantum physics.  Of crucial importance here is the difference
between the quantum and classical domains: the quantum analogue of
the satisfying assignment, namely the ground state, can exhibit
extremely intricate multi-particle entanglement. This additional
player in the game makes borrowing results from the classical domain
to the quantum domain extremely challenging, cf. the wide open major
open problem of whether a quantum analogue of PCP holds
\cite{ref:Aha09b}; it also opens up completely new directions of
research regarding the entanglement properties of ground states of
local Hamiltonians.

General quantum states require $2^n$ complex numbers to describe. 
One of the major goals of quantum Hamiltonian complexity is to 
derive bounds on the entanglement exhibited in ground states of 
interesting classes of local Hamiltonians; the purpose of those
bounds and restrictions on the entanglement is to lead to an
efficient description and analysis of ground states in cases of
interest.  There is a beautiful sequence of papers using structures
called tensor networks, with special cases such as MPS
\cite{ref:Ost95, ref:Ost97, ref:Vid04a, ref:Ver04b}, PEPS
\cite{ref:Ver04a}, TN \cite{ref:Shi06}, and MERA \cite{ref:Vid07a},
which provide such efficient descriptions in certain cases. 

Area laws constitute one of the most important tools for bounding 
entanglement in such systems.  Consider the interaction graph
(hypergraph) associated with a local Hamiltonian -- it has a vertex
for each particle and an edge for each term of the Hamiltonian.
Intuitively and very roughly, an area law says that entanglement is
local in this interaction graph in the following sense: consider a
subset of particles $L$. Then the entanglement between $L$ and
$\bar{L}$ in the ground state is locally ``concentrated'' along the
edges between $L$ and $\bar{L}$; more precisely, the area law states
that the entanglement entropy across the cut is big-Oh of the number
of edges crossing between $L$ and $\bar{L}$.  This is clearly a very
strong restriction on the entropy, which in the general case would
be of order of the number of particles (nodes) in $L$.  Proving area
laws for typical classes of Hamiltonians is thus a holy grail in
quantum Hamiltonian complexity. 

A few years ago, in a seminal paper~\cite{ref:Has07}, Hastings
proved that the area law holds for 1D systems (i.e., when the
interaction graph is a path), for gapped Hamiltonian -- that is,
Hamiltonians whose overall spectral gap is of order $\orderof{1}$. 
In this case, the area law says that ground state entanglement
across any contiguous cut is bounded by a constant. From this, one
can deduce that the ground state of such systems can be described
efficiently (by an MPS of polynomial bond dimension -- see
\Ref{ref:Has07}).  The question of whether ground states in two
and higher dimensions obey an area law is still wide open.

Hastings' proof of the 1D area law, and many other proofs related to
entanglement and correlations in ground states, use sophisticated
analytic methods.  Perhaps the most important of those is the famous
Lieb-Robinson bound (LR bound) \cite{ref:Lie72, ref:Has04}, which
bounds the velocity at which disturbances propagate in quantum local
systems; Fourier analysis, and other techniques are important players
too.  These analytic tools constitute a major barrier for a fuller
participation by computer scientists in this important aspect of
Hamiltonian complexity.  Also, these analytic techniques seem to
inherently involve the dynamics of the system in time, according to
the Hamiltonian. However, purely from an aesthetics point of view,
it should be possible to explain kinematic results about the ground
state without resorting to dynamical arguments (which is what the LR
bound is). Or, in other words, without adding the extra dimension of
time to the problem.  In addition, the kinematic problem seems, at
least on the surface, to be of a combinatorial nature, thereby
suggesting a combinatorial solution. 

In this paper we introduce a combinatorial tool to tackle the above
mentioned problems, and, in particular, to get a handle on 
correlations and entanglement in ground states of local
Hamiltonians. This is a simple, basic version of the detectability
lemma of \Ref{ref:Aha09b}. We demonstrate that when the system is
frustration-free, many of the results that rely on the traditional
analytic tools can be obtained in a much simpler, direct and
intuitive way using this tool; we argue that the detectability lemma
in this form constitutes a basic tool for the study of local
Hamiltonians and their ground states from a computational point of
view.  

Our starting point is the Detectability Lemma (DL) introduced in
\Ref{ref:Aha09b}.  There, the motivation for the DL was quite
specific: to help translate classical results about CSPs to quantum
results about local Hamiltonians. It was used to prove a quantum
analog of gap amplification (a component of Dinur's proof of the PCP
theorem \cite{ref:Din07}). The DL made it possible to sensibly make
a statement of the form ``If the ground state energy is at least $k$
then the probability that it violates at least $ck$ terms of the
Hamiltonian is bounded below by a constant''. The DL of
\Ref{ref:Aha09b} holds under the mild assumption (which is
essentially true in most interesting cases) that each particle
participates in a bounded number of terms of the Hamiltonian, and
therefore the terms of the Hamiltonian can be partitioned into a
constant number of layers, each consisting of terms acting on
disjoint sets of particles. \Ref{ref:Aha09b} also required an
additional technical assumption, that the number of distinct types
of terms of the Hamiltonian are bounded. 

Here, we reformulate the DL and put it in a much broader and basic 
context. Our reformulation of the DL asks the following question:
consider a gapped frustration-free local Hamiltonian $H =
\sum_{i=1}^m H_i$ with $0 \leq H_i \leq 1$. i.e., the ground energy
of $H$ is $0$, and the spectral gap is $\epsilon = \orderof{1}$. The
frustration-free assumption means that the ground state minimizes
the energy of \emph{every} local term, so no term is ``frustrated''.
Can we approximate the projection $\gs $ on the ground state,
$\ket{\Omega}$, by a ``local'' operator? Such a local approximation
would be extremely useful, as it would enable deducing local
properties of the ground states such as area laws and decay of
correlations.  Indeed, such an approximation of a projection on the
ground state is essentially what is done by the traditional analytic
tools that use the LR bound, as we explain in \Sec{sec:LR}. The
approximation offered by the DL, however, has more of a
combinatorial flavor, and is therefore much easier to handle.

A natural first guess of such a local approximation of $\gs $ is the
positive semi-definite operator $G\EqDef(\Id - \frac{1}{m}H)$, where
$m$ is the number of terms in the Hamiltonian. $G$ fixes the ground
state, and shrinks all the orthogonal space to it by a factor; 
however the shrinkage is very limited, by a factor of
$(1-\epsilon/m)$. To get a good approximation, one would need to
apply this operator polynomially many times, and by this we would
lose the locality of the operator. Indeed, the expression
$(\Id-H/m)^m$ contains products of $m$ overlapping terms whose
overall support is of the order the size of the system. Our
challenge is therefore to get a local operator that preserves the
ground state but shrinks the orthogonal subspace by a constant
factor, rather than by $\epsilon/m$. For simplicity of presentation
in the introduction, let us consider the simplest scenario, in which
the particles are set on a 1D chain, and the interactions are two
local. Denote by $P_i$ the projection on the ground state of the
terms $H_i$.  Notice that the terms in the Hamiltonian can be
partitioned into two layers, the even and odd terms, each acting on
disjoint sets of particles (see \Fig{fig:2layers}); 

\begin{figure}
  \begin{center}
    \includegraphics[scale=0.8]{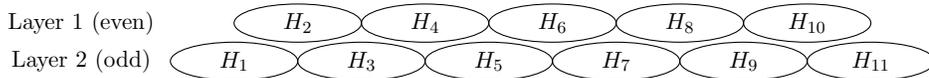}
  \end{center}
  \caption{ Am illustration of a 1D system of two-local, nearest
  neighbors, interactions. The local terms ($H_1, H_2, \ldots$) can
  be arranged in two layers (even and odd), such that the terms in
  each layer do not overlap. \label{fig:2layers}}
\end{figure}

Denote by $\Pi_{odd}$ the product of the projections on the ground
spaces of all odd terms $P_1,P_3,...$ and by $\Pi_{even}$ the
product for the even terms.  Then the operator
$A\EqDef\Pi_{odd}\Pi_{even}$ is the ``local'' operator we want.  The
DL states: 

\begin{lem}[ Detectability Lemma (DL) in $1D$]
\label{lem:detect2}
  Let $A\EqDef\Pi_{odd}\Pi_{even}$, and let $\Hi'$ be the orthogonal
  complement of the ground space. Then
  \begin{align}
  \label{eq:shrink}
      \norm{A|_{\Hi '}}
        \leq \frac{1}{(\epsilon/2 + 1)^{1/3}} \ .
  \end{align}
\end{lem}

The DL says that the application of $A$ to any vector moves the
vector closer to the ground state of $H$ by cutting down the mass in
the orthogonal subspace by a constant factor.  This implies that
$\gs $, the projection into the ground space of $H$, can be
approximated to within exponentially good precision by applying the
operator $A$ $\ell$ times: $\gs  = A^{\ell} + e^{-\orderof{\ell}}$.  

Let us explain why this operator is indeed
``local''. When $A$ is applied $\ell$ times to some local
perturbation $B$ that acts on the ground state $\ket{\Omega}$, there
is a pyramid-shaped ``causality cone'' of projections that is
defined by $B$.  These are simply all terms which are
graph-connected to the operator $B$ (see \Fig{fig:light-cone}).  All
the projections outside that cone commute with $B$ and can therefore
be absorbed in the ground state (since $P_i\ket{\Omega} =
\ket{\Omega}$), leaving us with a \emph{local} operator of support
of size $\orderof{\ell}$. Effectively, $A^\ell$ acts non-trivially
only on a region of width $\orderof{\ell}$, when applied to
$B\ket{\Omega}$.

\begin{figure}
  \begin{center}
    \includegraphics[scale=1]{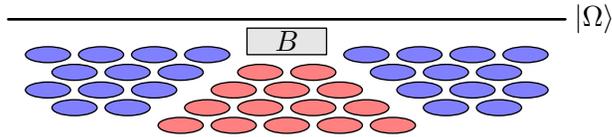}
  \end{center}
  \caption{An illustration of the expression $A^\ell B\ket{\Omega}$
  in a 1D system. $B$ is a local perturbation, applied to the ground
  state $\ket{\Omega}$. The local terms underneath it correspond to
  the $P_i$ projections in $A^\ell$. The pink terms are the
  projections inside the causality cone of $B$. These terms are
  graph connected to $B$, and generally do not commute with it. The
  blue terms are outside the causality cone, and can therefore
  commute with $B$ and be absorbed by $\ket{\Omega}$.
  \label{fig:light-cone}}
\end{figure}

We give here a new simple proof of this reformulation of the DL, in
the process dropping the assumption of \Ref{ref:Aha09b} about the
number of distinct types of terms of the Hamiltonian. 

The proof hinges on the following observation, which we refer to as
the \emph{norm-energy trade off}.  Assume by contradiction that $A$ 
does not move a vector $\ket{\psi}$, which is orthogonal to the
ground state, very much. Then $A\ket{\psi}$ must be very close to
the range of each $P_i$, but since the range of $P_i$ is the null
space of the local term $H_i$, this means that the energy
$\bra{\psi}A^\dagger H_i A\ket{\psi}$ must be small.  However, on
the other hand, the sum of those energies must be larger than
$\epsilon$, since $A\ket{\psi}$ is orthogonal to the ground space;
this implies that the shrinkage must be quite significant, providing
an upper bound on the norm of the vector $A\ket{\psi}$. 

However, the above argument is not sufficiently strong.  Since there
are $m$ terms $H_i$, the energy contribution of each term can be as
small as $\epsilon/m$; this will lead to a factor of $\epsilon/m$ in
the lemma, which is not strong enough.  The key point is that the
energy-norm trade off can be applied locally, using the tensorial
structure of $A$; we break the movement of $\ket{\psi}$ to
$A\ket{\psi}$ into disjoint sequential steps and then relate the
contributions to the energy of {\it each} of the terms $H_i$ with
the shrinkage resulting from each step; one might suspect that 
entanglement could prevent such an analysis in which shrinkage
accumulates but the point is exactly that the local structure of the
problem allows this accumulation to happen. Think very 
simplistically of the state $a\ket{00000}+b\ket{11111}$ subjected to
the local terms $H_i=\ket{1}\bra{1}_i$. A projection of this state 
on the ground state of $H_i$ for any one of qubits $i$, results in a
shrinkage by a factor of $|a|^2$, but once one projection is applied
in one location, the shrinkage is exhausted and no more shrinkage is
to be gained by a projection in another location.  That this
entanglement related phenomenon does not happen in the DL scenario
is due to the locality of the operators involved; it highlights that
the way the state $A\ket{\psi}$ can be entangled is severely
limited.  

We demonstrate the applicability of this reformulation of the DL by
providing significant simplifications of the proof of Hastings' area
law in 1D \cite{ref:Has07}, using the DL in two key points, 
bypassing completely the analytic methods. By this we hope to make
this important result accessible to a wider audience, as well as possibly
extendable to higher dimensions. The outline of the proof still
follows that of Hastings, but now becomes much easier to understand;
we defer the explanation of how the proof goes and how the
detectability lemma enters the picture to \Sec{sec:area}. 

To give another example, we provide a one page, very simple proof of Hastings'
celebrated result that the correlations in the ground states of gapped
Hamiltonians decay exponentially with the distance \cite{ref:Has04}.
Unlike the area law, this applies to $d$-dimensional grids for any constant 
$d$.  More precisely, consider two observables $A$ and $B$ that are local and
act on sets of particles that are of distance $\ell$ on the grid;
the decay of correlations means that the expectation value of their
product is almost as that of the product of their expectation, up to
an error which decays exponentially in $\ell$.

We mention that at first sight, one might connect the exponential decay
of correlations to an intuition that entanglement between a region
$L$ and its surrounding is ``located'' only close to the boundary of
$L$, and thus scales like the area rather than like the volume.
Though an appealing intuition, such an implication of exponential
decay of correlation to area laws is not known, and indeed quantum
expanders provide a counter-example to such a na\"ive connection
\cite{ref:Has07b}.

In both of those proofs, the DL replaces a combination of the
Lieb-Robinson bound with other analytic tools; this works of course
only when the DL is applicable, namely, for the rich case of
frustration-free Hamiltonians. The restriction to frustration-free
Hamiltonians may seem quite strong. We note, however, that there are
various frustration-free systems that are interesting from a physics
and a computational points of view, such as the ferromagnetic XXZ
model, the AKLT model \cite{ref:Aff87}, and stabilizer codes such as
the Toric code \cite{ref:Kit03}. In addition, many of the quantum
phenomenon in quantum Hamiltonian complexity are revealed already in
the context of frustration-free Hamiltonians, and the major open
problems in this area (e.g., quantum PCP and 2D area law) are wide
open already for this case. Much is to be learned from studying
frustration-free Hamiltonians, before we proceed to the more general
case; it seems that the simpler combinatorial nature of the DL in
this case might provide a new handle to those questions, and there
are reasons to believe that a proof of an area law for
frustration-free systems might be extendable to the general case. 

To illustrate how exactly the is DL related to the analytic methods,
we start our more technical discussion with a toy application
comparing the usage of the LR bound to the alternative route offered
by the DL, in \Sec{sec:LR}.  

{~}

\noindent\textbf{Related work and further directions:}\\ The DL
seems to be connected to various diverse scientific areas. The
connections to the LR bound and other analytic tools used in
condensed matter physics are discussed extensively in \Sec{sec:LR}; 
one other connection is to view of the DL operator $A$ as a special
instance of the general \emph{Method of Alternating Projections
(MAP)}, that was first studied by von Neumann \cite{ref:Neu50}. In
that method one applies a fixed sequence of projections in order to
approach the intersection subspace.  In the general setting, the
projections are not assumed to be local, nor the Hilbert space is
assumed to be of finite dimension. In recent results
\cite{ref:Bad10}, the convergence rate is given as a function of the
\emph{Fridriechs angle}, which is not easily related to a physical
quantity. The DL, on the other hand, is a MAP under the special
assumption that the projections are local, associated with a
frustration-free $k$-local Hamiltonian, with a convergence rate that
is given as a function of the spectral gap. It would be interesting
to see if more insight can be derived from these connections.  

Recently, much attention was given to a quantum algorithm which,
given a local Hamiltonian, uses a process involving random 
\emph{measurements} of the energies of the local terms to approach
the ground state efficiently (for certain cases) \cite{ref:Ver09,
ref:Far09}. The algorithm discussed in those papers carries
similarities to the situation we are handling here, despite the fact
that measurements are applied rather than projections, and also that
the terms are chosen randomly, rather than in some fixed order.  It
seems that the DL lemma, and the energy-norm trade off, could
potentially be useful also for the analysis of such algorithms.  In
particular, it would be very interesting to see a version of the
detectability lemma which applies for the case in which the terms
are chosen randomly. 

As discussed above, it is a wide open question to apply the
combinatorial tools presented in this paper to the major open
problems of quantum PCP and area laws in dimensions higher than $1$,
as well as to many other basic open questions in quantum Hamiltonian
complexity.

{~}

\noindent\textbf{Paper organization:}\\ 
We start with notations and preliminaries in \Sec{sec:not}, and then
proceed to the statement and proof of the DL in \Sec{sec:det}. In
\Sec{sec:LR}, we provide the example comparing the LR bound approach
to the DL one. We then proceed to the area law proof in
\Sec{sec:area}, and conclude with the one page proof of the
exponential decay in \Sec{sec:exp}.

\section{Notations and Preliminaries}
\label{sec:not}

We consider a $k$-local Hamiltonian $H$ acting on
$\Hi=(\BBC^d)^{\otimes n}$, the space of $n$ particles of dimension
$d$. $H=\sum_i H_i$ where each $H_i$ is a non-negative and bounded
operator that acts non-trivially on a constant number of $k$ qubits
(hence the term local Hamiltonian).  We assume that $H$ has a
ground space of energy $0$, which must therefore also be a common
zero eigenspace of all terms $H_i$. This means that $H$ is
frustration free.  We also assume that $H$ is ``gapped'', meaning
that its lowest eigenvalue is 0 (the ground energy) and all the next
are equal or larger than some constant $\epsilon>0$. We denote by
$\Hi' \subset \Hi$ the orthogonal complement ground space of $H$.
Thus $\Hi'$ is an invariant subspace for $H$, and 
\begin{align} 
\label{eq:gapcondition}
  H|_{\Hi'} \geq \epsilon \Id \ . 
\end{align}

Most of these assumptions, except for perhaps the frustration-free
assumption, are very often used in condensed matter physics.

Throughout this paper we further assume that the $H_i$'s are
projections, and hence would be denoted by $Q_i$. We define $P_i$ to
be the projection on the ground space of $Q_i$, $P_i\EqDef\Id-Q_i$. 
The assumption that $H$ is made of projections is not actually a
restriction because we can reduce any frustration-free, bounded and
gapped system into that case. Specifically, for $H=\sum_i H_i$ with
$\norm{H_i}\le K$, and a spectral gap $\tau>0$, we first add an
appropriate constant to each $H_i$ such that their ground energy is
0. Then for every $i$ we define $Q_i$ as the projection into the
space where the energy of $H_i$ is greater than 0 and $P_i\EqDef
\Id-Q_i$ as the projection to the ground space of $H_i$. Finally, we
define the auxiliary Hamiltonian $H'=\sum_i Q_i$.  This system is
frustration free because the original ground states would also be
ground states in $H'$ with a vanishing energy. Moreover, for any
state $\ket{\psi_\perp}\in\Hi'$ and every $H_i$,
\begin{align*}
  \bra{\psi_\perp}H_i\ket{\psi_\perp} =
  \bra{\psi_\perp}Q_iH_iQ_i\ket{\psi_\perp} \le
    K\bra{\psi_\perp}Q_i\ket{\psi_\perp} \ ,
\end{align*}
and therefore the gap in $H'$ is $\epsilon \ge \tau/K$. It follows
that all of our results can be applied to bounded frustration-free
Hamiltonians by replacing the gap $\epsilon$ in DL with the scaled
version $\tau/K$. 

Given a state $\ket{\phi}$ and a partition of the qubits to two 
non intersecting sets, $R$ and $L$, with corresponding Hilbert spaces
$\Hi_L, \Hi_R$, we can consider the Schmidt decomposition of the
state along this cut: $\ket{\phi} = \sum_j \alpha_j
\ket{L_j}\otimes\ket{R_j}$.  Here $\alpha_1\ge \alpha_2\ge \ldots$
are the \emph{Schmidt coefficients}. Their squares are equal to the
non-zero eigenvalues of the reduced density matrices to either
side of the cut $\rho_L(\phi)$ and $\rho_R(\phi)$, which we denote
by $\lambda_1 \geq \lambda_2 \geq \cdots$.  The \emph{Schmidt rank}
of $\ket{\phi}$ is then the number of non-zero eigenvalues
$\lambda_j$ (or Schmidt coefficients $\alpha_j$), and the
\emph{entanglement entropy} is the entropy of the set
$\{\lambda_i\}$, or, equivalently, the von Neumann entropy of the
matrix $\rho_L(\phi)$. A straightforward corollary of the
Eckart-Young theorem \cite{ref:Eck36} is then that the truncated
Schmidt decomposition provides the best approximation to a vector in
the following sense:
\begin{fact}  
\label{f:rankapprox} 
  Let $\ket{\phi}$ be a vector on $\Hi_L \otimes \Hi_R$, and let
  $\lambda_1\ge \lambda_2\ge\ldots$ be the eigenvalues of its
  reduced density matrix.  The largest inner product between
  $\ket{\phi}$ and a norm one vector with Schmidt rank $r$ is
  $\sqrt{\sum_{j=1}^r\lambda_j}$.
\end{fact}

\section{The detectability lemma: A new proof}
\label{sec:det}

For clarity of presentation, we will prove the DL in the case stated
in the introduction: where the particles are set on a line and the
local terms are two-local involving nearest neighbors.  This proof
contains all the necessary ingredients for the proof of the more
general DL in the case where the Hamiltonian has $k$ local terms
that can be partitioned into $g$ layers; we make the
precise statement of the more general case at the end of this
section.

We begin with a simple lemma that quantifies the norm-energy
trade-off in the simple case of two projections $X,Y$: we show that
if the application of $XY$ does not move a vector very much then the
energy of that vector with respect to $\Id-Y$ must be small:

\begin{lem}
\label{lem:normenergy}
 Given arbitrary projections $X,Y$ and $\ket{v}$ of norm 1, if
 $\norm{XYv}^2=1-\epsilon$ then
 \begin{align}
 \label{eq:normenergy}
   \norm{(\Id-Y)XY v}^2 \leq \epsilon(1-\epsilon) \ .
 \end{align}
\end{lem}
The proof is given in the Appendix. Let us now proceed to prove the
detectability lemma.
\begin{proof}[\ of \Lem{lem:detect2}] \ \\
  
  Suppose $\ket{\psi}\in \Hi'$ is a norm 1 state that is orthogonal
  to the ground space, and define $\ket{\phi}\EqDef A\ket{\psi}$. 
  Notice that for every ground state $\ket{\Omega}$,
  $\bra{\Omega}A\ket{\psi}=0$ and so $\ket{\phi}$ is orthogonal to
  the ground space. We would like to show that
  \begin{align}
    \label{eq:phi-ineq}
    \norm{\phi}\le\frac{1}{(\epsilon/2+1)^{1/3}} \ .
  \end{align}
  We will find both a lower and upper bounds for the energy of
  $\ket{\phi}$, $\bra{\phi}H\ket{\phi}$, which
  will give us an inequality for $\norm{\phi}^2$, from which
  \Eq{eq:phi-ineq} will follow. 
  
  The lower bound is straightforward since $\ket{\phi}$ is
  orthogonal to the ground state, and so 
  \begin{align}
  \label{eq:lowerbound}
    \bra{\phi}H\ket{\phi} \ge \epsilon \norm{\phi}^2 \ .
  \end{align}

  We shall now upper bound the energy $\bra{\phi}H\ket{\phi}$ by
  carefully upper bounding the contributions of the individual terms
  $\bra{\phi}Q_i\ket{\phi}$.  We begin by noting that these terms
  are equal to $0$ for $i$ odd since $A = \Pi_{odd} \Pi_{even}$ and
  $Q_i \Pi_{odd}=0$ for any odd $i$ (recall that $\Pi_{even},
  \Pi_{odd}$ are products of the projections $P_i=\Id-Q_i$).  We now
  want to bound the contributions coming from the even terms.

  For this purpose we present $A$ in a convenient form, by
  reordering its terms.  We call the triplet product of projections
  $(P_1 P_3 P_2), (P_5 P_7 P_6), \ldots$ \emph{pyramids}, and denote
  them by $\Delta_i= P_{4i-3} P_{4i-1} P_{4i-2}$; The remaining
  terms are combined to the operator $R\EqDef P_4 P_8 \dots$. See
  \Fig{fig:2layers-with-pyr} for an illustration of this structure
  in 1D. Notice that by just using the fact that $P_i$ and $P_j$
  commute when $i$ and $j$ are not consecutive, we can write: 
  \begin{align*}
    A= \Delta _1 \Delta_2 \dots \Delta_{m} R \ ,
  \end{align*} 
  where $m$ is the number of pyramids which is approximately
  $n/4$.  

  \begin{figure}
    \begin{center}
      \includegraphics[scale=0.8]{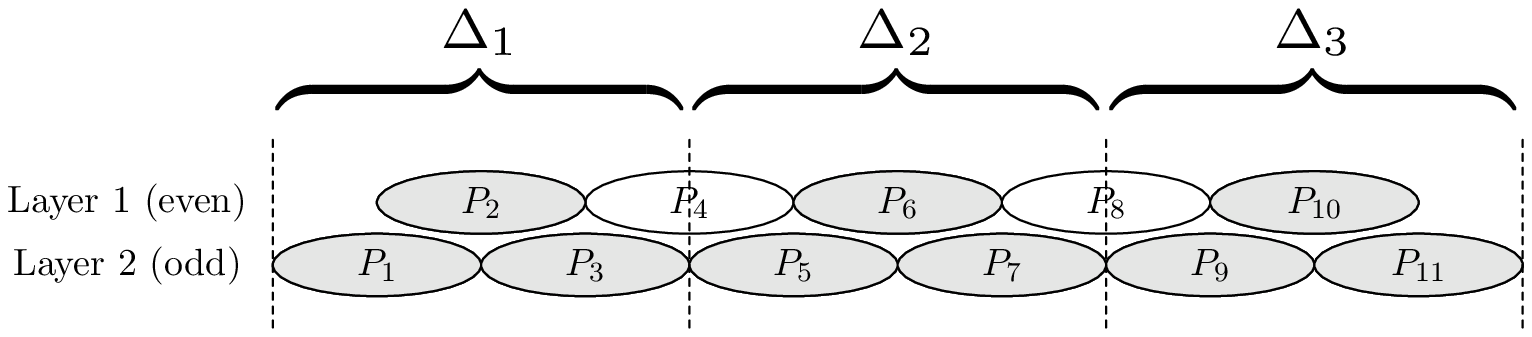}        
    \end{center}
    \caption{ \label{fig:2layers-with-pyr}}
  \end{figure}

  We will use this reordering to bound the energy contribution of
  the terms $Q_2, Q_6 , \dots$; a symmetric argument will bound the
  remaining even terms $Q_4,Q_8, \ldots$ etc.  The energy contribution of
  $Q_{4i-2}$ will be related to the amount of movement produced by
  the $\Delta _i$ portion of of the operator $A$.

  The key point in providing this bound is this.  We view the
  transformation of $\ket{\psi}\rightarrow\ket{A \psi}= \ket{\phi}$
  as a series of steps given by the application of the pyramids
  $\Delta _i$.  Specifically, letting $\ket{v_i}\EqDef
  \Delta_i\Delta_{i+1}\cdots\Delta_{m}R\ket{\psi}$, we consider the
  transformation $\ket{\psi} \rightarrow R \ket{\psi} \rightarrow
  \ket{v_{m}} \rightarrow \ket{v_{m-1}}\rightarrow \dots \rightarrow
  \ket{v_1}=A\ket{\psi}$.  The square of the norm of the first
  state, after applying $R$, is $a_{m}\EqDef\norm{R \psi}^2$.  Let
  $a_i\EqDef \norm{v_i}^2/\norm{v_{i+1}}^2$ be the ``shrinkage'' of
  the square norm (or movement) resulting from the application of
  the $i$th pyramid, for $1\leq i<m$.
 
  We shall now prove, using \Lem{lem:normenergy}, that the shrinkage
  $a_i$ is related to the energy of the operator $Q$ at the top of
  the same pyramid $\Delta_i$:
  \begin{align*}
    \bra{\phi} Q_{4i-2} \ket{\phi} \leq 1-a_i \ .
  \end{align*}
  We write
  \begin{align*}
    \bra{\phi}Q_{4i-2}\ket{\phi} 
      &= \norm{(\Id-P_{4i-2}) A \psi}^2
       = \norm{\Delta_1\cdots\Delta_{i-1} 
         (\Id- P_{4i-2})\Delta_i v_{i+1}}^2 \\
      &\leq \norm{(\Id-P_{4i-2} )\Delta_i v_{i+1}}^2 \ ,
  \end{align*} 
  and recall that $\Delta_i=P_{4i-3}P_{4i-1}P_{4i-2}$. We can
  therefore apply \Lem{lem:normenergy} to $(\Id-P_{4i-2})\Delta_i
  \frac{v_{i-1}}{\norm{v_{i-1}}}$ (with $Y=P_{4i-2}$ and
  $X=P_{4i-3}P_{4i-1}$), and conclude that
  $\norm{(\Id-P_{4i-2})\Delta_i v_{i-1}}^2 \leq (1-\norm{\Delta_i
  v_{i-1}}^2/\norm{v_{i-1}}^2) \norm{v_{i-1}}^2 = (1-a_i)
  \norm{v_{i-1}}^2 \leq (1-a_i)$. Consequently $\bra{\phi}
  Q_{4i-2}\ket{\phi} \le 1-a_i$. 
  
  Using this upper bound gives an upper bound for the energy
  contribution for $Q_i$, $i \in \{2,6,10, \dots \}$: 
  \begin{align*}
    \bra{\phi} (Q_2 + Q_6  + \ldots)\ket{\phi}
      \le \sum_i (1-a_i) \ ,
  \end{align*}
  with the constraint on the norm $\norm{\phi}^2 \le\prod
  a_i$.  The right hand side is maximized when all the $a_i$s are
  equal to each other, i.e., $a_i= \norm{\phi}^{\frac{2}{m}}$, and
  therefore we are left with an upper bound of the energy coming
  from $Q_2 + Q_6 + \dots$ as:
  \begin{align*}
        \bra{\phi} (Q_2 + Q_6  + \ldots)\ket{\phi}
          \le m\left[1-\norm{\phi}^{2/m}\right] 
          \le \frac{1-\norm{\phi}^2}{\norm{\phi}} \ ,
  \end{align*}
  where the last inequality follows from the fact\footnote{This
  inequality can be easily verified by noticing that $f_m(x)\EqDef
  m\sqrt{x}\left[1-x^{1/m}\right]+x$ is equal to 1 for $x=1$
  and has a non-negative derivative for $x\in[0,1]$ and $m\ge 1$.} that for every
  $x\in[0,1]$, we have $m\left[1-x^{1/m}\right] \le
  \frac{1-x}{\sqrt{x}}$.
  
  For the energy of $Q_4 + Q_8 + \ldots$, a similar
  decomposition to $A= (P_3 P_5 P_4)(P_7P _9 P_8) \cdots (P_2 P_6
  \cdots )$ can be made, thereby upper bounding the energy by
  $2\frac{1-\norm{\phi}^2}{\norm{\phi}}$. Combining the energy upper and lower
  bounds we therefore get
  \begin{align}
    \epsilon\norm{\phi}^2 \le \bra{\phi}H\ket{\phi} \le
    2\frac{1-\norm{\phi}^2}{\norm{\phi}} \ ,
  \end{align}
  and so
  \begin{align}
    \frac{\epsilon}{2}\norm{\phi}^3 \le 1-\norm{\phi}^2 \le 1-\norm{\phi}^3 \ ,
  \end{align}
  which  gives
  \begin{align}
    \norm{\phi} \le \frac{1}{\left( \epsilon/2 + 1\right)^{1/3}} \ .
  \end{align}
\end{proof}

The above proof can be easily generalized to other geometries. In
the general case, in accordance with \Sec{sec:not}, we assume we
have a $k$-local, frustration-free Hamiltonian $H=\sum_i Q_i$ that
is made of projections and has a spectral gap $\epsilon>0$. We
further assume that each particle participates in a constant number
of projections, and therefore the $Q_i$ can be partitioned into a
constant number of $g$ layers; each layer is made of projections
that do not intersect each other and are therefore commuting.

Then for each layer we define the projection $\Pi_i$ as the product
of all $P_j=\Id-Q_j$ that are in the layer, and define the DL operator $A$
by
\begin{align}
  \label{def:gen-A}
  A \EqDef \Pi_g \cdots \Pi_1 \ .
\end{align}
Finally, we define $f(k,g)$ to be the number of \emph{sets} of
pyramids that are necessary to estimate the energy contribution of
all the $Q_i$ terms. In the 1D case that we proved, we had
$f(k,g)=2$, because only the even layer contributed energy and we
needed two sets of pyramids to cover that layer. In the general case it
is easy to see that $f(g,k)$ can be crudely bounded by $f(g,k) \le
(g-1)k^g$.

Using the above definitions, the general DL is
\begin{lem}[The detectability lemma]
\label{lem:gen-DL}
  Consider the local Hamiltonian system that is described above. Then
  \begin{align}
  \label{eq:gen-DL}
    \norm{A|_{\Hi'}} \le \frac{1}{\big[\epsilon/f(k,g) +1\big]^{1/3}} \ .
  \end{align}  
\end{lem}

\section{Comparing the Lieb-Robinson bound approach and the 
  detectability lemma approach}
\label{sec:LR}

In this section, we compare the DL with a standard method used in 
many of the seminal results in quantum Hamiltonian complexity, such
as Hastings' areas law for 1D gapped systems \cite{ref:Has07}, and
Hastings' exponential decay of correlations proof \cite{ref:Has04}.
The method combines the use of the Lieb-Robinson bound (LR bound),
with a Fourier analysis and the existence of a gap, to reveal the
locality properties of the ground state. More specifically, the
method uses these tools to approximate expressions that involve the
projection operator to the ground state, $\gs $, by local operators.

To understand how this is done,
let us concentrate on a simple example of locality in the ground
state, and derive it using both the DL and the LR bound.

We focus on $\gs $, the projection on the ground space of $H$.  On
the surface, this projector seems very far from being local in any
sense. Nevertheless, in gapped systems it does possess some locality
properties that are crucial to the analysis of correlations and
entanglement in the ground state. To see this, one standartly
considers an approximation of $\gs $ by another operator: 
\begin{align}
  \label{eq:Pq}
  P_q \EqDef \frac{1}{\sqrt{2\pi q}} \int\! dt \, e^{-t^2/2q} e^{-iHt} \ .
\end{align}
Here $q$ is a free parameter to be chosen as appropriate.  For an
eigenvector $\ket{E}$ of $H$ with eigenvalue $E$, we have
$P_q\ket{E} = e^{-qE^2/2}\ket{E}$. Consequently, if the system has a
constant spectral gap $\epsilon>0$, $P_q$ indeed approximates $\gs $
well: 
\begin{align*}
  \norm{P_q - \gs } \le e^{-q\epsilon^2/2} \ .
\end{align*}

We now want to argue regarding the local nature of $P_q$ in various
contexts. Let us illustrate the LR bound approach with a simple
example: consider the expression $\gs  B \ket{\Omega}\approx P_q
B\ket{\Omega}$, where $\ket{\Omega}$ is a ground state of the system
with zero energy, and $B$ is some local perturbation. It is easy to
see that 
\begin{align}
\label{eq:Pq-action}
  P_q B \ket{\Omega}
   = \frac{1}{\sqrt{2\pi q}}\int\! dt\, e^{-t^2/2q} e^{-iHt} B e^{iHt}
     \ket{\Omega}
   =\frac{1}{\sqrt{2\pi q}}\int\! dt\, e^{-t^2/2q} B(t)\ket{\Omega} \ .
\end{align}
$B(t)\EqDef e^{-iHt}Be^{iHt}$ is the time evolution of the
perturbation $B$. The key point is now to use the famous LR bound, to
approximate it by a ``local'' operator, i.e., an operator which acts
only on the ``neighborhood'' of the particles on which $B$ acts. The
following is an immediate corollary of the original LR bound, which
we omit for sake of brevity. The full statement of the LR bound,
together with the proof of this corollary can be found in
\Ref{ref:Has10}.

\begin{theorem}[Lieb-Robinson bound (LR bound) , adapted \cite{ref:Has10}]
  Given a local Hamiltonian $H=\sum_i H_i$ on $n$ particles, there
  exists a constant velocity $v$ s.t. $B(t)$ can be approximated by an
  operator denoted $B_\ell(t)$ whose support is inside a ball of
  radius $\ell=vt$ around the support of $B$,  s.t.  
  \begin{align}
    \label{eq:approx}
    \norm{B(t)-B_\ell(t)} \le \norm{B}\cdot e^{-\orderof{\ell}} \ .
  \end{align}
\end{theorem}

Given a length scale $\ell>0$, we may now set $q=\ell$ in
\Eq{eq:Pq-action} and obtain
\begin{align*}
  \gs  B\ket{\Omega} \approx P_\ell B\ket{\Omega} \approx
    \frac{1}{\sqrt{2\pi\ell}}\int_{|t|\le \ell}\! dt\, 
      e^{-t^2/2\ell} B(t)\ket{\Omega}
   \approx \frac{1}{\sqrt{2\pi\ell}}\int_{|t|\le \ell}\! dt\, 
      e^{-t^2/2\ell} B_\ell(t)\ket{\Omega} \ .
\end{align*}
In the above series of approximations $\approx$ implies an
approximation of up to an error of $e^{-\orderof{\ell}}$. The 1st
approximation follows from the assumption of the constant gap and
\Eq{eq:Pq}. The 2nd approximation is due to the exponential decay of
the filter function $e^{-t^2/2\ell}$, and the 3rd is due to the LR
bound. We therefore get an exponentially (in $\ell$) good
approximation to $\gs B$ in the expression $\gs B\ket{\Omega}$ by an
operator which is $\ell$-local. 

Let us now derive the same result for the frustration-free case
using the DL. First, we approximate the ground space projection
$\gs $ by applying the DL operator $A$ for $m$ times. By
\Eq{eq:gen-DL}, $A$ leaves the ground space invariant while
shrinking the orthogonal space by a constant factor. Therefore
\begin{align}
  \label{eq:central}
  \gs  = A^m + e^{-\orderof{m}} \ .
\end{align}
We now write
\begin{align*}
  \gs  B\ket{\Omega} = A^mB\ket{\Omega} + \norm{B}\cdot
    e^{-\orderof{m}} \ ,
\end{align*}
and consider the expression $A^mB\ket{\Omega}$. By assumption, the
system is frustration free, and therefore every local projection
operator $P_i$ that appears in $A$ leaves $\ket{\Omega}$ invariant:
$P_i\ket{\Omega}=\ket{\Omega}$. We now consider the ``causality
cone'' of projections in $A^m$ that are defined by $B$. These are
simply the projections that are graph-connected to $B$ when all the
projections in $A^m$ are arranged in consecutive $gm$ layers (see
\Fig{fig:light-cone}). The main observation is that all the
projections outside this causality cone commute with $B$, and can
therefore be absorbed by $\ket{\Omega}$. We are therefore left only
with the projections of the causality cone, whose support size
$\ell$ is proportional to $m$. In other words, just as in the LR
bound method, we found an exponentially (in $\ell$) good
approximation to $\gs B$ in the expression $\gs B\ket{\Omega}$ by an
operator which is $\ell$-local. 

This kind of reasoning, with appropriate modifications, is used in
both the 1D area-law and the exponential decay of correlations, that
are presented in the following sections.


\section{ The area law in 1D using the detectability lemma}
\label{sec:area}  

Throughout this section, we let $H = \sum Q_i$ be a 2-local
frustration-free 1D Hamiltonian that is made of projections $Q_i$
acting on particles of dimension $d$. Assume that $H$ has a unique
ground state $\ket{\Omega}$ and a spectral gap $\epsilon$, and set
$\delta\EqDef 1-(1+\epsilon/2)^{-1/3}$ in accordance with the
shrinking exponent of the DL, in \Eq{eq:shrink}. We notice that in
the interesting limit $\epsilon \to 0$, we have $\delta \simeq
\epsilon/6$, and generally, for for $\epsilon\in [0,1]$, have
$\epsilon/8 \le \delta \le \epsilon/6$. In order to keep the
presentation simple, we shall assume throughout this section that
$\epsilon \le 1$, and prove the following version of a one
dimensional area law:

\begin{theorem}[Area Law for frustration free Hamiltonians in $1D$]
\label{thm:al} 
  For any contiguous cut along the chain, the
  entanglement entropy of the ground state $\ket{\Omega}$ across the
  cut is bounded by a constant which depends on the dimensionality
  of the particles $d$ and on the spectral gap $\epsilon$;
  specifically,
  \begin{align}
    \label{eq:final-bound}
    S \le \frac{10}{\delta}d^{4/\delta}(\ln d)^2 \ .
  \end{align}
  
\end{theorem}

The proof relies on two main lemmas. The first shows that for any
cut along the line, there is a product state $\ket{\phi_1} \otimes
\ket{\phi_2}$ that has a constant inner product with the ground state:

\begin{lem}[Constant overlap with a product state]
\label{lem:overlap} 
  For every cut, there is a product state
  $\ket{\phi_1}\otimes\ket{\phi_2}$ such that
  $|\braket{\phi_1\otimes\phi_2}{\Omega}| \geq \mu\EqDef
  d^{-\ell}(1-\delta)^{\ell_0/4}$, with $\ell_0\EqDef d^{4/\delta}$.
\end{lem}

The second lemma shows that if there exists a product state with
a constant overlap with the ground state $\ket{\Omega}$, then
$\ket{\Omega}$ has finite entanglement entropy:

\begin{lem}[Constant overlap with a product state implies finite entropy]
\label{lem:overlaptoarealaw} 
  If for some cut there exists a product state $\ket{\phi_1}\otimes\ket{\phi_2}$
  such that $|\braket{\phi_1\otimes\phi_2}{\Omega}|\ge \mu$, then the
  entanglement entropy of $\ket{\Omega}$ across the cut is bounded by
  \begin{align}
  \label{eq:entropy-ub}
    S  \le \frac{3}{\delta}(\ln \frac{1}{\mu^2\delta} + 2)\ln d \ .
  \end{align}
\end{lem}

Theorem \ref{thm:al} then follows easily by combining the two lemmas
and using the facts that $\delta\ge \mu$ and $\ln \frac{1}{\mu} \ge
6$.  We prove \Lem{lem:overlap} in \Sec{sec:overlap} and
\Lem{lem:overlaptoarealaw} in \Sec{sec:overlaptoarealaw}.

%
\subsection{Constant overlap implies finite entropy 
  (proof of \Lem{lem:overlaptoarealaw})} 
\label{sec:overlaptoarealaw}

In this section we prove Lemma \ref{lem:overlaptoarealaw}.  The DL
is clearly the right tool for the task, since it provides a ``local''
operator that can be repeatedly applied to the promised product
state $\ket{\phi_1}\otimes\ket{\phi_2}$ without increasing its
entanglement rank much, while exponentially decreasing its distance 
from the ground state. 

The only thing that is not entirely clear is how to get a constant
bound on the entanglement entropy of the ground state, since a
straightforward argument would mean applying the operator
non-constant number of times to get arbitrarily close to the ground
state. The key is to observe that after $\ell$ applications of the
DL we get a state with a bounded Schmidt rank that is close to the
ground state, and by Fact~\ref{f:rankapprox}, this gives us a bound
on the sum of the largest $d^{2\ell}$ Schmidt coefficients \emph{of
the ground state}. With these bounds we can find a pessimistic
\emph{constant} upper bound on the entanglement entropy. We can now
proceed to the more detailed proof. 

Consider then a cut in the line between the particles $i_0$ and
$i_0+1$, and let $Q_{i_0}$ be the local term in $H$ that involves
$i_0, i_0+1$. Assume that along that cut, the product state
$\ket{\phi_0}\EqDef \ket{\phi_1}\otimes\ket{\phi_2}$ has a constant
projection $\mu$ on the ground state $\ket{\Omega}$:
\begin{align}
 \ket{\phi_0} = \mu\ket{\Omega} + \ket{w} \ ,
\end{align}
where $\ket{w}\in \Hi'$, and $\norm{w} = (1-\mu^2)^{1/2}<1$.
We now apply the operator DL operator $A$  $\ell$ times on
$\ket{\phi_0}$. We obtain
\begin{align*}
 A^\ell\ket{\phi_0} = \mu\ket{\Omega} + \ket{w^{(\ell)}} \ ,
\end{align*}
where $\ket{w^{(\ell)}}\in \Hi'$ and
$\norm{w^{(\ell)}}\le (1-\delta)^\ell$ \ .
Let $\ket{v_\ell}$ be the normalized version of $A^\ell\ket{\phi_0}$. Then 
\begin{align}
\label{eq:vl-inprod}
 |\braket{v_\ell}{\Omega}| \ge 
   \frac{\mu}{\sqrt{\mu^2 + (1-\delta)^{2\ell}}}
   =1-\orderof{(1-\delta)^{2\ell}} \ .
\end{align}
This means $\ket{v_\ell}$ are exponentially close to the ground
state, as a function of $\ell$. 

How entangled are those states? 
We notice that at each application of $A$, the entanglement
rank of the state can only increase by a multiplicative factor of
$d^2$: for every $i\neq i_0$, the projection term $P_i$ in $A$
works entirely left to the cut or entirely right to the cut, thereby
not increasing the Schmidt rank of the state. The only projection in
$A$ that may increase the rank is $P_{i_0}$, and as it is a 2-local
projection that works on $d$-dimensional particles, it can at most
increase it by a factor of $d^2$.\footnote{Given a vector $\ket{v}=
\sum_{j=1}^r \ket{L}_j \otimes \ket{R_j}$ of Schmidt rank $r$ and
$i<i_0$, $P_i\ket{v}=\sum_{j=1}^r
\big(P_i\ket{L_j}\big)\otimes\ket{R_j}$ and thus has Schmidt rank
bounded by $r$ (the symmetric argument shows the same result for
$i>i_0$).  For $i=i_0$, we can decompose $P_{i_0}= \sum_{k=1}^{d^2}
X_k \otimes Y_k$ with $X_k$ acting on the $i_0$'th particle and
$Y_k$ acting on the $(i_0+1)$'th particle.  Consequently
$P_{i_0}\ket{v} = \sum_{j=1}^r \sum_{k=1}^{d^2}
(X_k\ket{L_j})\otimes (Y_k\ket{R_j})$ has Schmidt rank no larger
than $rd^2$.} Consequently, the Schmidt rank of $\ket{v_\ell}$
is at most $d^{2\ell}$.

We have therefore obtained a family of states $\{ \ket{v_\ell} \}$
with Schmidt ranks bounded by $d^{2\ell}$, which are closer and
closer to $\ket{\Omega}$. Then using Fact~\ref{f:rankapprox},
together with \Eq{eq:vl-inprod}, it follows that the eigenvalues of
the reduced density matrix of $\ket{\Omega}$ along the cut,
$\lambda_1\ge \lambda_2\ge\ldots$, must satisfy 
\begin{align}
  \frac{\mu}{\sqrt{\mu^2 + (1-\delta)^{2\ell}}} \le
  |\braket{v_\ell}{\Omega}| \le \left(\sum_{j=1}^{d^{2\ell}}
  \lambda_j\right)^{1/2} \ ,
\end{align}
which implies the following series of
inequalities:
\begin{align}
 \text{For every $\ell\ge 1$:} \quad
 \sum_{j\ge d^{2\ell}+1} \!\!\lambda_j
   \le \frac{1}{\mu^2}(1-\delta)^{2\ell}\ .
\label{eq:schmidtub}
\end{align}

From here, the desired upper bound on the entropy can be deduced by
choosing the distribution of maximal entropy which still satisfies
the inequalities in \Eq{eq:dist-constraints}. The following lemma,
whose proof can be found in \App{sec:S-upperbound}, gives one such
bound:
\begin{lem}
\label{lem:steps}
 Consider a probability distribution $\{\lambda_j\}$ whose values
 are ordered in a non-increasing fashion, $\lambda_1\ge
 \lambda_2\ge\ldots$, and let $D\ge 2$ be an integer and $K\ge 1,
 0<\theta<1$ some constants such that
 \begin{align}
 \label{eq:dist-constraints}
   \text{for every $\ell\ge 1$:} \quad
   \sum_{j\ge D^\ell+1} \!\!\lambda_j
   \leq K\theta^\ell \ .
 \end{align}
 Then the entropy of $\{\lambda_j\}$ is upper bounded by
 \begin{align}
 \label{eq:S-bound}
   S \le  
     3\left[\frac{\ln \frac{K}{1-\theta} + 1}{\ln(1/\theta)} +
     2\right]\ln D \ .
 \end{align}
\end{lem}

Substituting $D=d^2, \theta=(1-\delta)^2$ and $K=\mu^{-2}$ gives
\begin{align*}
  S \le 6\left[\frac{\ln \frac{1}{\mu^2[1-(1-\delta)^2]} +
  1}{2\ln[1/(1-\delta)]} + 2\right]\ln d \ .
\end{align*}
But $1-(1-\delta)^2 \ge \delta$ and $\ln[1/(1-\delta)] \ge \delta$,
and so
\begin{align*}
  S \le 6\left[\frac{\ln \frac{1}{\mu^2\delta} + 1}{2\delta} 
    + 2\right]\ln d
  \le \frac{3}{\delta}(\ln \frac{1}{\mu^2\delta} + 2)\ln d \ ,
\end{align*}
where the last inequality follows from the assumption that $\delta
\le \epsilon/6 \le 1/6$.

\subsection{A product state having constant overlap with $\ket{\Omega}$
  (proof of \Lem{lem:overlap})} 
\label{sec:overlap}

The obvious candidate for a tensor product state with a constant
overlap with $\ket{\Omega}$ is the mixed state $\rho_L \otimes
\rho_R$, where $\rho_L$ is the reduced density matrix of
$\ket{\Omega}$ to the left of the cut, and $\rho_R$ is the reduced
density matrix to the right. 

Let us assume for contradiction that the overlap between
$\ket{\Omega}$ and $\rho_L \otimes \rho_R$, and in fact with any
tensor product state along a certain cut, is less than
$(1-\delta)^{\ell/4}$ for some sufficiently large constant $\ell$.  
If the overlap is small, then there is a measurement that
distinguishes $\ket{\Omega}$ from $\rho_L \otimes \rho_R$ with
probability of at least $1-(1-\delta)^{\ell/2}$; this is simply the projection
on the ground state, $\gs $. 

The challenge is to show that there is a local such measurement, i.e., a 
measurement confined to a local window, which distinguishes these
two states almost as well. Using the DL we shall now find such
local measurement that distinguishes with a slightly worse probability
$1-2(1-\delta)^{\ell/2}$.

Let us denote by $\rho ^{\ell}_L$ (respectively $\rho ^{\ell}_R$)
the reduced density matrix of $\ket{\Omega}$ restricted to the
$\ell$ particles to the left (respectively right) of the cut. Also,
let $\rho ^{2\ell}$ be $\rho=\ket{\Omega}\bra{\Omega}$ restricted to
the $2 \ell$ particles, $\ell$ on each side of the cut. We refer to
the state $\rho ^{\ell}_L\otimes\rho ^{\ell}_R$ as the
``disentangled'' version of the state $\rho ^{2\ell}$.  The
following lemma shows that under the assumption that $\ket{\Omega}$
has low overlap with every product state $\ket{\phi_1} \otimes
\ket{\phi_2}$ (along a given cut), there exists a measurement
confined to the window of $2\ell$ particles around the cut, that
with high probability distinguishes $\rho ^{2\ell}$ from $\rho
^{\ell}_L \otimes \rho ^{\ell} _R$. 

\begin{lem}[Existence of a distinguishing measurement]
\label{lem:measurement}   
  Assuming that the overlap of the ground state with any product
  state satisfies $|\braket{\phi_1\otimes\phi_2}{\Omega}|\leq
  (1-\delta)^{\ell/4}$, there is a measurement that distinguishes
  $\rho^{2\ell}$ from $\rho^\ell_L \otimes \rho^\ell_R$ with
  probability $1 - 2(1-\delta)^{\ell/2}$.
\end{lem} 

The DL ensures that by applying the layers one by one, we converge
to the projection on the ground state quickly, and it is this
projection that is exactly the distinguishing measurement we want to
approximate.  We can thus apply $A$ only $\ell/2$ times,
approximating the projection on the ground space; now, following the
intuition explained in the introduction (and in the example of
\Sec{sec:LR}), only the causality cone of the cut should be used in this
measurement, and the rest of the operators in those layers are
swallowed by the state being measured; this amounts to a measurement
which is restricted to the $-\ell,\ell$ interval and still
distinguishes well enough. The detailed proof can be found in the
appendix.  

The fact that such a measurement exists, distinguishing the original
state confined to the $2\ell$ window from its ``disentangled''
version, with high probability, must somehow indicate that there is
a lot of entanglement along the cut, whose disentanglement caused
this distinguishability.  This can be made precise using an
information-theoretical argument: 

\begin{lem}[Distinguishing measurement implies large difference in
entropies]
\label{lem:A}
  If there is a measurement that distinguishes $\rho ^{2\ell}$ from
  $\rho^{\ell}_L\otimes\rho^{\ell}_R$ with probability of at least
  $1 - 2(1-\delta)^{\ell/2}$, then
  \begin{align*}
    S(\rho^{\ell}_L) + S(\rho^{\ell}_R) 
      - S( \rho^{2\ell} ) \ge \frac{\delta}{2}\ell-1 \ .
  \end{align*}
\end{lem}

The lemma implies that the entropy in $S(\rho^\ell_L) +
S(\rho^\ell_R)$ is significantly larger than $S(\rho^{2\ell})$,
implying that disentangling along the cut has introduced a lot of 
new entropy. The proof is simple, based on relative entropy; 
essentially, all it uses is the fact that a measurement that
distinguishes with high probability two states, implies high
relative entropy between the results of the measurements.  Once
again, details can be found in the appendix. 

To finish the proof of \Lem{lem:overlap}, we now need to derive a
contradiction.  Denote by $S(2\ell)$ the value of $S(\rho^{2\ell})$,
with $2\ell$ being the segment centered around the cut that provides
our contradictory assumption (namely, that any tensor product state
has less than $(1-\delta)^{\ell/4}$ inner product with the ground
state). Under these conditions, Lemma \ref{lem:measurement} applies,
and hence also the conditions of Lemma \ref{lem:A} apply to this
segment.  Applying \Lem{lem:A} we conclude that $S(2\ell) \leq 2
S(\ell)-\frac{\delta}{2}\ell+1$.  We now want to recursively apply
this inequality, for the $\ell$ long segments on both sides of the
cut, and then for the $\ell/2$-long segments within those segments,
and so on.  The problem is that the cuts now move to different
locations within the $2\ell$ long window, and so our assumption no
longer applies for these cuts. However, if the inner product state
with any tensor product state is small along the original cut, it
can be shown to be quite small also along near-by cuts, and so all
the above arguments can be applied for those cuts too. This can be
formalized in the following claim, whose easy proof can once again
be found in the appendix: 

\begin{claim}  
\label{cl:cut} 
  If $|\braket{\phi_1 \otimes \phi_2}{\Omega}| \leq \mu$ for all
  product states across cut $(k,k+1)$ then
  $|\braket{\chi_1\otimes\chi_2}{\Omega}| \leq \mu d^\ell$ for all
  product states across any cut $(k+j, k+j+1)$ with $-\ell \leq j
  \leq \ell$.
\end{claim}

 
We therefore assume by contradiction that the inner product along a
given cut is smaller than $\mu$, such that $\mu d^\ell=
(1-\delta)^{\ell/4}$, and so along all the cuts in the $2\ell$
window we have that the inner product of the states is at most
$(1-\delta)^{\ell/4}$, and hence our assumptions apply.  We can
therefore use the same argument recursively.  Since $S(1)\leq
\ln(d)$, we get (for $\ell$ a power of 2), $S(\ell) \leq \ell \ln(d)
- \frac{\delta}{2} \ell \log_2\ell +\ell \leq \ell (\ln(d)+1) -
\frac{\delta}{2}\ell \log_2\ell $.  Choosing $\ell$ such that
$\frac{\delta}{2}\log_2\ell \geq \ln d+1$ makes $S(\ell)<0$ thus
giving a contradiction. Using the fact that $d\ge 2$, this can be achieved by
$\ell=d^{4/\delta}$ since $d^{4/\delta} \ge 2^{\frac{2(\ln d + 1)}{\delta}}$. 

\section{Exponential decay of correlations} \label{sec:exp}

Consider now a Hamiltonian $H=\sum_i Q_i$ which is $k$-local and set
on a $d$-dimensional grid. Once again, we assume that $Q_i$ are
projections, and that $H$ is frustration free with a unique ground
state $\ket{\Omega}$ (i.e., $Q_i\ket{\Omega}=0$), and a spectral gap
$\epsilon>0$. 
We wish to show

\begin{theorem} {\bf Decay of Correlations in ground states of gapped 
Hamiltonians on a $d$-Dim grid:}

Consider a setting as described above.  Let $X,Y$ be two local
observables whose distance on the grid from each other is $m$.
Denote $\bar{X}\EqDef\bra{\Omega}X\ket{\Omega}$,
$\bar{Y}\EqDef\bra{\Omega}Y\ket{\Omega}$.  Then

\begin{align}
 |\bra{\Omega} (X-\bar{X})(Y-\bar{Y})\ket{\Omega}|
   =  |\bra{\Omega} XY\ket{\Omega} - \bar{X}\bar{Y}| \le
   \norm{X}\cdot\norm{Y}\cdot e^{-\orderof{m}}  \ .
\end{align}
\end{theorem}

\begin{proof}

Let us now consider two operators: $P_{in}, P_{out}$: $P_{in}$ is
defined by applying the DL $\ell$ times to $Y$ and discarding all
projections outside the causality cone of $Y$.  $\ell$ is chosen such
that the resulting cone will not overlap with $X$ (see
\Fig{fig:in-and-out}). Therefore $\ell\propto m$, with the
proportionality constant that is a geometrical factor. $P_{out}$ is
the complement of $P_{in}$, i.e., it the layers that one get by
applying the DL $m$ times, but with a ``hole'' where the causality
cone of $Y$ is. Together, we have $P_{in}\cdot P_{out}= A^\ell$ -- See
\Fig{fig:in-and-out} for an illustration in 1D.

\begin{figure}
\begin{center}
 \includegraphics[scale=1]{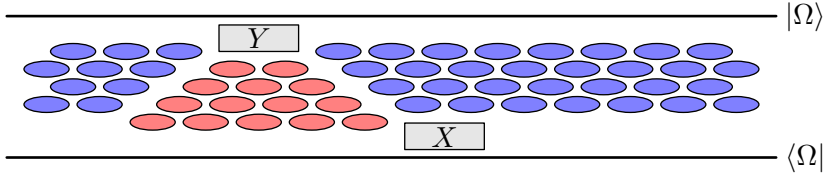}
 \caption{ An illustration of the statement $\bra{\Omega} X Y \ket{\Omega}
   = \bra{\Omega} X P_{out}P_{in} Y \ket{\Omega}$. The operator
   $P_{out}$ is drawn in blue color and $P_{in}$ is in red. Note
   that the number of projection layers is proportional to the
   distance between $X$ and $Y$.
 \label{fig:in-and-out}
 }
\end{center}
\end{figure}

$P_{in}, P_{out}$ leave the ground-state invariant. In addition, they
commute with $X$ and $Y$ respectively, hence
\begin{align*}
 \bra{\Omega}X = \bra{\Omega}XP_{in} \ , \qquad
 Y\ket{\Omega} = P_{out}Y\ket{\Omega} \ ,
\end{align*}
and therefore
\begin{align*}
 \bra{\Omega} XY \ket{\Omega}
   = \bra{\Omega} X P_{in}P_{out} Y\ket{\Omega}
   = \bra{\Omega} X A^\ell Y \ket{\Omega} \ .
\end{align*}
We now recall that $A^\ell$ is in fact an approximation of the ground
state projection $\gs$ (see \Eq{eq:central} in \Sec{sec:LR}),
\begin{align*}
 A^\ell = \gs  + e^{-\orderof{\ell}} = \gs  + e^{-\orderof{m}} \ ,
\end{align*}
and so
\begin{align*}
 \big|\bra{\Omega} XY \ket{\Omega}- \bra{\Omega} X \gs
 Y\ket{\Omega}\big| \le \norm{X}\cdot\norm{Y}\cdot e^{-\orderof{m}} \ .
\end{align*}
Assuming that the ground state is unique,
$\gs =\ketbra{\Omega}{\Omega}$, and therefore
\begin{align*}
 \big|\bra{\Omega} XY \ket{\Omega} - \bar{X}\bar{Y}\big|
  \le \norm{X}\cdot\norm{Y}\cdot e^{-\orderof{m}} \ .
\end{align*}

\end{proof}

\section{Acknowledgments}
\label{sec:Acknowledgements}

We are grateful to Matt Hastings, Tobias Osborne and Bruno
Nachtergaele for inspiring discussions about the above and related
topics.

Itai Arad acknowledges support by Julia Kempe's ERC Starting Grant
QUCO and Julia Kempe's Individual Research Grant of the Israeli
Science Foundation 759/07.

\vskip -.6cm
\bibliographystyle{ieeetr}

{~}

{
\small 
\bibliography{QC}
}

%
%


\appendix

\section{Proof of Norm-Energy trade-off, \Lem{lem:normenergy}}

\begin{proof}[\ of \Lem{lem:normenergy}] \ \\

 Set $\ket{w} \EqDef Y\ket{v}/\norm{Yv}$. Then
 \begin{align*}
   \norm{(\Id-Y)XYv}^2 = \norm{Yv}^2\cdot\norm{(\Id-Y)Xw}^2 \ .
 \end{align*}
 By definition, $\ket{w}$ is a normalized vector inside the support
 of $Y$ and therefore for every vector $\ket{\psi}$, we have
 $\norm{(\Id-Y)\psi} \le \norm{(\Id-\ketbra{w}{w})\psi}$.
 Plugging this to the equality above, we find
 \begin{align*}
   \norm{(\Id-Y)XYv}^2 &\le \norm{Yv}^2\cdot\norm{(\Id-\ketbra{w}{w})Xw}^2
     = \norm{Yv}^2\cdot\bra{w}X(\Id-\ketbra{w}{w})X\ket{w} \\
     &= \norm{Yv}^2\cdot\norm{Xw}^2\cdot(1-\norm{Xw}^2) \\
     &=\norm{XYv}^2\cdot(1-\norm{Xw}^2) \\
     &\le \norm{XYv}^2\cdot(1-\norm{XYv}^2) =(1-\epsilon)\epsilon
     \ ,
 \end{align*}
 where the last inequality follows from the fact that $\norm{Xw}\ge
 \norm{XYv}$.
\end{proof}


\section{Upper bound on the Entropy}
\label{sec:S-upperbound}

In this section we prove \Lem{lem:steps}, which is  required to finish the 
proof of Lemma \ref{lem:overlaptoarealaw}. 

\begin{proof}\ \\

Call the set of weights $\{ \lambda_j\}$ for $D^\ell+1 \le j\le
D^{\ell+1}$ the $\ell$'th block. Then 
the constraints in \Eq{eq:schmidtub} imply that for every
block $\ell\ge 1$, 
\begin{align}
  \sum_{j=D^\ell+1}^{D^{\ell+1}} \lambda_j 
    \le K\theta^\ell \ .
\end{align}
Obviously, by reshuffling the mass within a block we maintain the
constraints. Moreover, it is straight forward to see that the
entropy contribution of every block is maximized when all the
weights in it are equal. The maximal distribution is therefore a
steps function, which satisfies:
\begin{align}
  \text{in block $\ell$}, \quad \lambda_j \le
  \frac{K\theta^\ell}{D^{\ell+1}-D^\ell} 
  = \frac{K}{D-1}(\theta/D)^\ell \ .
\end{align}
We now define $\ell_0$ to be the first block for which
$K\theta^\ell \le \frac{1}{e}(1-\theta)\theta$:
\begin{align}
\label{eq:equiv}
   \frac{\ln \frac{eK}{\theta(1-\theta)}}{\ln(1/\theta)} &\le \ell_0
     \le \frac{\ln \frac{eK}{\theta(1-\theta)}}{\ln(1/\theta)} + 1
     = \frac{\ln \frac{K}{1-\theta} + 1}{\ln(1/\theta)} + 2 \ .
\end{align}

We will bound the maximal entropy by bounding the entropy
contribution of blocks up to (and including) $\ell_0-1$ and blocks
from $\ell_0$ onwards. The first is easy, as there are
$D^{\ell_0}$ weights in the low blocks:
\begin{align}
  S_I \le  \ell_0\ln D \ .
\end{align}

In the high blocks, $\lambda_j \le \frac{1}{e}(1-\theta) \le 1/e$, so we
can use the monotonicity of the function $-\lambda\ln\lambda$ in the
$(0:1/e]$ range to bound the entropy by
\begin{align}
  S_{II} &\le -\sum_{\ell\ge\ell_0}  
    K\theta^\ell \ln [\frac{K}{D-1}(\theta/D)^\ell] 
    \le \sum_{\ell\ge\ell_0}  
    K\theta^\ell \ln (D^{\ell+1}) \\
   &=\frac{K\theta^{\ell_0}\ln D}{1-\theta}
     \left(\ell_0 + \frac{1}{1-\theta}\right) \\
   &\le \ln D \left(\ell_0 + \frac{1}{1-\theta}\right) \ ,
\end{align}
where the second equality follows from standard geometric sums
identities, and the last inequality follows from the definition of
$\ell_0$. Next, looking at the lower bound of $\ell_0$ in
\Eq{eq:equiv}, it takes standard calculus to verify that $\ell_0 \ge
\frac{1 + \ln\frac{1}{1-\theta}}{\ln (1/\theta)} + 1\ge
\frac{1}{1-\theta}$, and so $S_{II}
\le 2\ell_0\ln D$, and 
\begin{align*}
  S = S_I + S_{II} \le 3\ell_0\ln D \ .
\end{align*}
Plugging the upper bound of $\ell_0$ from \Eq{eq:equiv}, we
get \Eq{eq:S-bound}.
\end{proof}


\section{Proof of Existence of Distinguishing Measurement, Lemma \ref{lem:measurement}}
\begin{proof}[\ of lemma \ref{lem:measurement}] \ \\

Let $\mathcal{Q} = \{ Q_i : Q_i \text{ acts only on particles in the
$2\ell$ interval}\}$.  Let $\Pi$ be a projection onto the
ground space of all the operators in $\mathcal{Q}$.  We will show
that $\{ \Pi, 1 - \Pi\}$ is the desired distinguishing measurement.

Clearly $\Tr( \Pi \rho^{2\ell}) =1$.  We would now like to prove
that $\Tr(\Pi \rho^\ell_L\otimes \rho^\ell_R)$ is at most
$2(1-\delta)^{\ell/2}$.

We start by considering, instead of $\Pi$, the applications of the
DL operator $A$.

We can write $\rho _L \otimes \rho_R$ as a convex combination of
rank 1 density matrices of product states of the form
$\ket{\phi_1}\otimes\ket{\phi_2}$.  By assumption, the overlap with
the ground state is $|\braket{\phi_1\otimes\phi_2}{\Omega}|\leq
(1-\delta)^{\ell/4}$. Therefore, since we assume a unique ground
state, $\ket{\phi_1\otimes\phi_2} = c\ket{\Omega} +
(1-c^2)^{1/2}\ket{\Omega_\perp}$, with $c\le (1-\delta)^{\ell/4}$
and $\ket{\Omega_\perp}$ perpendicular to the ground space. Then by
the DL, applying $A$ for $\ell/2$ times, we get
$\Tr(A^{\ell/2}\ketbra{\phi_1}{\phi_1}
\otimes\ketbra{\phi_2}{\phi_2})\leq [(1-\delta)^{\ell/4}]^2 +
(1-\delta)^{\ell/2} = 2(1-\delta)^{\ell/2}$, and this remains true
when we take convex combinations:
\begin{align*}
  \Tr(A^{\ell/2}\rho_L\otimes\rho_R)\le  2(1-\delta)^{\ell/2} \ .
\end{align*}

Thus $\Tr(A^{\ell/2}\rho)=1$, and $\Tr(A^{\ell/2}
\rho_L\otimes\rho_R)\le 2(1-\delta)^{\ell/2}$; this establishes that
applying $A$ for $\ell/2$ times distinguishes between $\rho$ and
$\rho_L\otimes\rho_R$ with the desired probability.  However, we
would like the measurement to be confined to a short interval. 
Since we know that $\Tr(\Pi\rho^{2\ell})=\Tr(\Pi\rho)=1$, the proof
will follow from showing that 
\begin{align*}
  \Tr(\Pi\rho_L^\ell \otimes\rho_R^\ell)
    = \Tr(\Pi\rho_L\otimes\rho_R) 
    = \Tr\big(\Pi A^{\ell/2} (\rho_L\otimes\rho_R)\big)
    \le \Tr\big(A^{\ell/2} (\rho_L\otimes\rho_R)\big) \ .
\end{align*}
The first equality holds since $\Pi$ only acts on the $2 \ell$
particles in $\rho _L ^{\ell} \otimes \rho_R^{\ell}$ and the last
inequality is trivial; it is the middle equality which uses the
structure of $A^{\ell/2}$. Indeed, let us write $A^{\ell/2}=A_M A_L
A_R$ where $A_M=\underbrace{\cdots(P_{k-2} P_k P_{k+2}) (P_{k-1}
P_{k+1}) (P_k)}_\text{$\ell$ groups}$ is a ``pyramid'' of terms
centered at the cut $k$, and $A_L$ and $A_R$ are the terms to the
left and right of the pyramid respectively, as in \Fig{fig:pyr}.
Then $\Pi A^{\ell/2}(\rho_L\otimes\rho_R)=\Pi A_M A_L
A_R(\rho_L\otimes\rho_R)$. But since we applied $A$ for exactly
$\ell/2$ times, every $P_i$ projection in $A_M$ is also in the
$2\ell$ window of $\Pi$. Therefore $A_M P_i=A_M$ and consequently
$\Pi A_M=\Pi$. 
Similarly, $A_L A_R(\rho_L\otimes\rho_R)=\rho_L\otimes\rho_R$, and
therefore $\Pi A^{\ell/2}(\rho_L\otimes\rho_R)=\Pi(\rho_L\otimes\rho_R)$, implying
the desired equality.

\begin{figure}
  \begin{center}
    \includegraphics[scale=1]{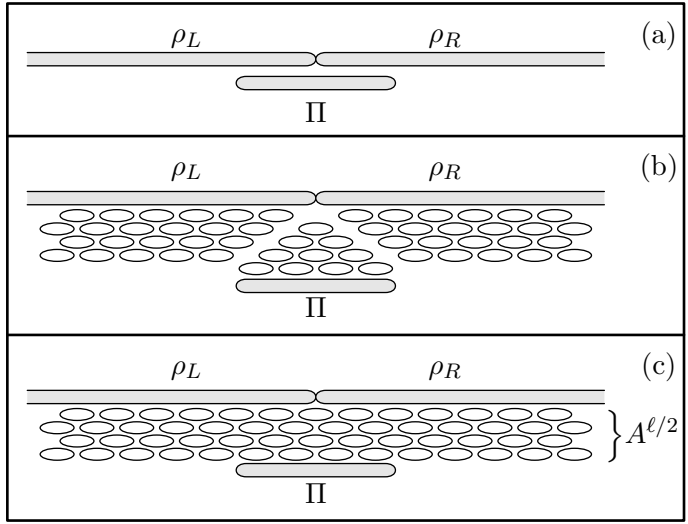}        
  \end{center}
  \caption{ \label{fig:pyr} An illustration of the identity $\Pi
  A^{\ell/2}(\rho_L\otimes\rho_R)=\Pi(\rho_L\otimes\rho_R)$. The
  left, right and middle sets of projections in Fig.~(b) are inside
  the invariant space of $\rho_L, \rho_R$ and $\Pi$ respectively.}
\end{figure}

\end{proof}


\section{Proof of Information Theoretical bound, Lemma \ref{lem:A}}

\begin{proof}[\ of \Lem{lem:A}] \ \\
  Let $X$ and $Y$ be $\{0,1\}$ random variables that result from
  applying the measurement $\Pi$ on $\rho$ and $\rho_L\otimes\rho_R$
  respectively.  Then by the Lindblad-Uhlmann theorem
  \cite{ref:Lin75,ref:Uhl77},
  \begin{align*}
    S(\rho_L^\ell) + S(\rho_R^\ell) - S(\rho^{2\ell}) 
      = S(\rho^{2\ell}||\rho_L^{\ell}\otimes\rho_R^{\ell}) 
      \geq S(X||Y) 
      = \sum_{i\in \{0,1\}} x_i \ln \frac{x_i}{y_i} \ .   
\end{align*}
  In this case, we have $X=1$ with probability $1$ and $Y$ is 1 with
  probability $\alpha \leq 2(1-\delta)^{\ell/2}$.  Thus using
  straight forward analysis, $\sum_i x_i
  \ln \frac{x_i}{y_i}= \ln(\frac{1}{\alpha}) \ge \ln(1/2) -
  \frac{\ell}{2}\ln(1-\delta)\ge \frac{\delta}{2}\ell-1$, and the result follows.
\end{proof} 

\section{Proof that close cuts behave similarly: Claim~\ref{cl:cut}}

\begin{proof}[\ of Claim \ref{cl:cut}] \ \\
  Assume for contradiction that $\ket{\chi_1}\otimes\ket{\chi_2}$ is
  a product state across the cut $(k+j, k+j +1)$ with $j>0$ such
  that $|\braket{\chi_1\otimes\chi_2}{\Omega}| > \mu d^\ell$. 
  Schmidt-decompose $\ket{\chi_1} =\sum_{i=1}^{d^\ell} \alpha_i
  \ket{L_i} \otimes\ket{R_i}$ where the cut is between the
  first $k$ particles and the $j$ particles between $k+1$ and $k+j$.
  By simple algebra, there exists at least one $i$ such that
  $|\braket{L_i\otimes R_i\otimes\chi_2}{\Omega}|>\mu$ which
  violates the hypothesis.
\end{proof}

\end{document}